\documentclass[journal]{IEEEtran}

\usepackage{amsfonts}
\usepackage{amsmath,amscd,multirow}
\usepackage{amssymb}
\usepackage{graphicx}%
\usepackage{color}
\usepackage{tikz}	
\usetikzlibrary{backgrounds,fit,decorations.pathreplacing}  

\newtheorem{theorem}{Theorem}

\newenvironment{proof}[1][Proof]{\noindent\textbf{#1.} }{\hfill \rule{0.6em}{0.6em}\\}

 \DeclareMathOperator{\tr}{Tr}

\newcommand{\ket}[1]{| #1 \rangle}
\newcommand{\proj}[1]{| #1 \rangle\!\langle #1 |}

\newcommand{\be}{{\mathbf e}}

\def\cA{{\cal A}}        

\def\cD{{\cal D}}        \def\cE{{\cal E}}

        \def\cK{{\cal K}}
\def\cL{{\cal L}}

\def\cM{{\cal M}}        \def\cN{{\cal N}}

\def\cU{{\cal U}}
        
\def\cX{{\cal X}}
\def\cY{{\cal Y}}        \def\cZ{{\cal Z}}

\def\0{{\mathbf{0}}}
\def\1{{\mathbf{1}}}
\def\2{{\mathbf{2}}}
\def\3{{\mathbf{3}}}
\def\4{{\mathbf{4}}}
\def\5{{\mathbf{5}}}
\def\6{{\mathbf{6}}}

\def\7{{\mathbf{7}}}
\def\8{{\mathbf{8}}}
\def\9{{\mathbf{9}}}

\def\fD{\mathfrak{D}}

\def\fS{\mathfrak{S}}

\def\fg{\mathfrak{g}}

\def\be{\begin{equation}}
\def\ee{\end{equation}}
\def\bea{\begin{eqnarray}}
\def\eea{\end{eqnarray}}

\DeclareMathOperator{\Tr}{Tr}

\begin{document}

\title{Channel Simulation and Coded Source Compression}

\author{Min-Hsiu Hsieh and Shun Watanabe\thanks{
This work was partially in Proceedings of IEEE Intl. Symp. Inf. Theory 2015, and in Proceedings of IEEE Information Theory Workshop 2015.

Min-Hsiu Hsieh is with the Centre for Quantum Computation \& Intelligent Systems, Faculty of Engineering and Information Technology,
University of Technology Sydney, Australia. Min-Hsiu Hsieh acknowledges the hospitality of the UTS-AMSS Joint Research Laboratory for Quantum Computation and Quantum Information Processing, Academy of Mathematics and Systems Science, Chinese Academy of Sciences, Beijing 100190, China.
(email: Min-Hsiu.Hsieh@uts.edu.au)

Shun Watanabe is with the Department of Computer and Information Sciences, Tokyo University of Agriculture and Technology, Japan.
(email: shunwata@cc.tuat.ac.jp).
}
}

\maketitle

\begin{abstract}
Coded source compression, also known as source compression with helpers, has been a major variant of distributed source compression, but has hitherto received little attention in the quantum regime. This work treats and solves the corresponding quantum coded source compression through an observation that connects coded source compression with channel simulation. First, we consider classical source coding with quantum side information where the quantum side information is observed by a helper and sent to the decoder via a classical channel. We derive a single-letter characterization of the achievable rate region for this problem. The direct coding theorem of our result is proved via the measurement compression theory of Winter, a quantum-to-classical channel simulation. Our result reveals that a helper's scheme which separately conducts a measurement and a compression is suboptimal, and measurement compression seems necessary to achieve the optimal rate region.  We then study coded source compression in the fully quantum regime, where two different scenarios are considered depending on the types of communication channels between the legitimate source and the receiver. We further allow entanglement assistance from the quantum helper in both scenarios. We characterize the involved quantum resources, and derive single-letter expressions of the achievable rate region. The direct coding proofs are based on well-known quantum protocols, the quantum state merging protocol and the fully quantum Slepian-Wolf protocol, together with the quantum reverse Shannon theorem. 

\end{abstract}


\section{Introduction}

Source coding normally refers to the information processing task that aims to reduce the redundancy exhibited when multiple copies of the same source are used.  In establishing information theory, Shannon demonstrated a fundamental result that source coding can be done in a \emph{lossless} fashion;  namely, the recovered source will be an exact replica of the original one when the number of copies of the source goes to infinity \cite{Shannon:1948wk}. If representing the source by a random variable $X$, with output space $\cX$ and distribution $p_X$, lossless source coding is possible if and only if the compression rate $R$ is above its Shannon entropy:
\begin{equation}
R\geq H(X),
\end{equation}
where $H(X):=\sum_{x\in\cX} -p_X(x) \log p_X(x)$.

Redundancy can also exist in scenarios where multiple copies of the source are shared by two or more parties that are far apart.  Compression in this particular setting is called \emph{distributed} source coding, which has been proven to be important in the internet era.  The goal is to minimise the information sent by each party so that the decoder can still recover the source faithfully.  Shannon's lossless source coding theorem can still be applied individually to each party. However, it has been discovered that a better source coding strategy exists if the sources between different parties are correlated. Denote by $X$ and $Y$ the sources held by the two distant parties, where the joint distribution is $P_{XY}$ and the output spaces are $\cX$ and $\cY$, respectively. 
Slepian and Wolf showed that lossless distributed source coding is possible when the compression rates $R_1$ and $R_2$ for the two parties satisfy \cite{Slepian:1973wj}:
\begin{align}
R_1 &\geq  H(X|Y),\\
R_2  &\geq H(Y|X),\\
R_1+R_2  &\geq H(XY),
\end{align}
where $H(X|Y)$ is the conditional Shannon entropy.  This theorem is now called the classical Slepian-Wolf theorem \cite{Slepian:1973wj}. In particular, when the source $Y$ is directly observed at the decoder, the problem is sometimes called source coding with (full) side information. 

Another commonly encountered scenario in a communication network is that a centralised server exists and its role is to coordinate all the information processing tasks, including the task of source coding, between the nodes in this network. Obviously, the role of the server is simply as a helper and it is not critical to reproduce the exact information communicated by the server. This is a major variant of distributed source compression, and proves to be important in the internet era. This slightly different scenario results in a completely different characterisation of the rate region, as observed by Wyner \cite{Wyner:1975iv} and Ahlswede-K\"orner \cite{Ahlswede:1975ea}. Consider that the receiver wants to recover the source $X$ with the assistance of the server (that we will call a helper from now on) holding $Y$, where the distribution is $P_{XY}$. Let $R_1$ and $R_2$ be the compression rates for the sender and the helper, respectively. Wyner and  Ahlswede-K\"orner showed that the optimal rate region for lossless source coding of $X$ with a classical helper $Y$ is the set of rate pairs $(R_1,R_2)$ such that 
\begin{align}
R_1 & \geq  H(X|U), \label{eq:classical-region-R1} \\
R_2 & \geq  I(U;Y), \label{eq:classical-region-R2} 
\end{align}
for some conditional distribution $p_{U|Y}(u|y)$, and $I(U;Y)$ is the classical mutual information between random variables $U$ and $Y$. When there is no constraint  on $R_2$ (i.e. $R_2$ can be as large as it can be), this problem reduces to source coding with (full) side information.

All of the above discussions are special cases of multi-terminal noiseless source coding problems \cite{Korner:1977dc, Csiszar:1980ba, TeSunHan:1980ci},  illustrated in Figure~\ref{fig:GMNSC}.  Consider $s$ dependent sources $X_1,X_2,\cdots,X_s$ and let $\fS=\{1,2,\cdots, s\}$. These $s$ sources are spatially separated, each observed by its own encoder $\cE_i$. There are $r$ decoders $\cD_1,\cdots,\cD_r$, and we denote $\fD=\{1,2,\cdots, r\}$.
Each decoder $\cD_k$ receives inputs from a pre-determined set $\fS_i\subset \fS$ of encoders and is required to output the $\fg(j)$-th source information $X_{\fg(j)}$ with vanishing probability of error as $n$ goes to infinity. The function $\fg: \fD \to \fS$, and we further assume that $\fg(j) \in\fS_j$, $\forall j\in\fD$. The target source $X_{\fg(j)}$ is called the \emph{primary} source for the $j$-th decoder $\cD_j$, while all other $\{X_p\}_{p\in\fS_j\backslash\fg(j)}$ are called the \emph{side} information for $\cD_j$. The achievable rate region $(R_1,R_2,\cdots, R_s)$ for this general multi-terminal noiseless source coding problem is derived in \cite{Csiszar:1980ba, TeSunHan:1980ci}. It is interesting to see that the Slepian-Wolf case corresponds to $\fS=\fD=\{1,2\}$, $\fS_j=\fS$, $\fg(j)=j$. The Wyner and Ahlswede-K\"orner case \cite{Wyner:1975iv, Ahlswede:1975ea} corresponds to $\fS=\{1,2\}$, $\fD=\{1\}$, $\fS_1=\{1,2\}$, $\fg(1)=1$. It also reduces to other instances of distributed source coding problems that fall beyond the main interests of this paper \cite{Wyner:1975iv, Sgarro:1977it, Korner:1977dc}.

\begin{figure}
\centerline{
    \begin{tikzpicture}[scale=1][very thick]
    \fontsize{10pt}{1} 
    \tikzstyle{halfnode} = [draw,fill=white,shape= underline,minimum size=1.0em]
    \tikzstyle{checknode} = [draw,fill=blue!10,shape= rectangle,minimum height=2.3em, minimum width=2.3em]
    \tikzstyle{checknode2} = [draw,fill=blue!10,shape= rectangle,minimum height=8em, minimum width=2em]
    \tikzstyle{variablenode} = [draw,fill=white, shape=circle,minimum size=0.05em]
    \node[checknode] (cn1) at (0,3.75) {${\cal E}_{1}$};
    \node[checknode] (cn2) at (0,2.75) {${\cal E}_{2}$};  
    \node[checknode] (cn3) at (0,1.75) {${\cal E}_{3}$};  
    \node[checknode] (cn4) at (0,-1.75) {${\cal E}_{s-1}$};  
    \node[checknode] (cn5) at (0,-2.75) {${\cal E}_{s}$};  
    \node (x0) at (0,0) {$\vdots$};
    \node (x1) at (-1.5,3.75) {$X_1^n$} ;
    \node (x2) at (-1.5,2.75) {$X_2^n$} ;
    \node (x3) at (-1.5,1.75) {$X_3^n$} ;
    \node (x4) at (-1.5,-1.75) {$X_{s-1}^n$} ;
    \node (x5) at (-1.5,-2.75) {$X_s^n$} ;
    \node[variablenode] (en1) at (3,3.75) {};
    \node[variablenode] (en2) at (3,2.75) {};  
    \node[variablenode] (en3) at (3,1.75) {};  
    \node[variablenode] (en4) at (3,-1.75) {};  
    \node[variablenode] (en5) at (3,-2.75) {};  
    \node (en0) at (3,0) {$\vdots$};
    \node[checknode] (dn1) at (5,3) {${\cal D}_{1}$};   
    \node[checknode] (dn2) at (5,1) {${\cal D}_{2}$};   
    \node[checknode] (dn3) at (5,-2) {${\cal D}_{r}$};   
    \node (xh1) at (6.5,3) {$\widehat{X}_{\fg(1)}^n$} ;
    \node (xh2) at (6.5,1) {$\widehat{X}_{\fg(2)}^n$} ;
    \node (xh3) at (6.5,-2) {$\widehat{X}_{\fg(r)}^n$} ;
    \node (xh0) at (6,0) {$\vdots$};
    \draw (cn1) -- (x1) (cn2) -- (x2) (cn3) -- (x3) (cn4) -- (x4) (cn5) -- (x5);
    \node (r1) at (1.5, 4) {$R_1$};
    \node (r2) at (1.5, 3) {$R_2$};
    \node (r3) at (1.5, 2) {$R_3$};
    \node (r4) at (1.5, -1.5) {$R_{s-1}$};
    \node (r5) at (1.5, -2.5) {$R_s$};
    \draw[thick,->] (cn1) -- (en1);
    \draw[thick,->] (cn2) -- (en2); 
    \draw[thick,->] (cn3) -- (en3);
   \draw[thick,->] (cn4) -- (en4);
   \draw[thick,->] (cn5) -- (en5);
    \draw[thick,->] (dn1) -- (xh1);
    \draw[thick,->] (dn2) -- (xh2); 
    \draw[thick,->] (dn3) -- (xh3); 
    \draw[thick,->] (en1) -- (dn1);
    \draw[thick,->] (en1) -- (dn2); 
    \draw[thick,->] (en2) -- (dn2);
    \draw[thick,->] (en3) -- (dn1);
    \draw[thick,->] (en3) -- (dn3);
    \draw[thick,->] (en4) -- (dn1);
    \draw[thick,->] (en5) -- (dn2);
    \draw[thick,->] (en5) -- (dn3);
   \end{tikzpicture}
 }

  \caption{
General multi-terminal noiseless source coding problem.
  }\label{fig:GMNSC}
\end{figure}
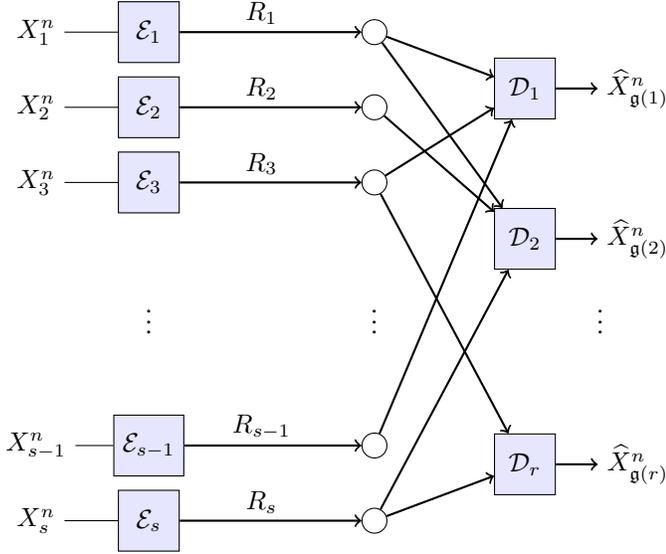

The problem of source coding,  when replacing classical sources with quantum sources, appears to be highly non-trivial in the first place\footnote{The first quantum source coding result  \cite{Schumacher:1995dg} took a much longer time to develop, considering that quantum theory began to evolve in the mid-1920s.}. The first quantum source coding theorem was established by Schumacher.  A quantum source $\rho_A$ can be losslessly compressed and decompressed if and only if the rate $R$ is above its von Neumann entropy\footnote{The subscript $A$ is a label to which the quantum system $\rho_A$ belongs.}:
\begin{equation}
R\geq H(A)_\rho, 
\end{equation}
where $H(A)_\rho:= -\tr \rho_A \log \rho_A$.

Schumacher's quantum source coding theorem bears a close resemblance to its classical counterpart. One will naturally expect that the same will hold true for the distributed source coding problem in the quantum regime. Consider that Alice, who has the quantum system $A$ of an entangled source $\rho_{AB}$, would like to merge her state to the distant party Bob.  The quantum distributed source coding theorem (also known as quantum state merging) aims to answer the optimal rate $R$ at which quantum information in system $A$ can be communicated faithfully to a party with quantum side information in system $B$. As it turns out, the optimal rate is given by the conditional von Neumann entropy $H(A|B)_\rho$, a quantum generalisation of classical conditional Shannon entropy. 
While the quantum formula to the distributed source coding problem is also of the form of conditional entropy, this result has a much deeper and profound impact in the theory of quantum information, as it marks a clear departure between classical and quantum information theory.   It is rather perplexing that the rate at which $R$ is quantified by the conditional entropy $H(A|B)_\rho$ can be negative. This major piece of the puzzle was resolved with the interpretation that if the rate is negative, the state can be merged, and in addition, the two parties will gain  $|H(A|B)_\rho|$ amount of entanglement for later quantum communication \cite{Horodecki:2005fv, Horodecki:2006hl, Dupuis:2014jz}. The distributed quantum source coding problem was later fully solved \cite{Abeyesinghe:2009ej,  Datta:2011vc} where the trade-off rate region between the quantum communication and the entanglement resource is derived.  The result is now called the fully quantum Slepian-Wolf theorem (FQSW).

Source coding with hybrid classical-quantum states $\rho_{XB}$, with $X$ representing a classical system and $B$ a quantum system, has also been considered in quantum information theory, and one of our main results falls into this category.  
In \cite{Devetak:2003kd}, Devetak and Winter considered classical source coding with quantum side information at the decoder, and showed that the optimal rate $R_1$ is given by $H(X|B)_\rho$. This result can be regarded as a classical-quantum version
of the source coding with (full) side information.

In this work, we first revisit the classical coded source compression \cite{Wyner:1975iv, Ahlswede:1975ea}. We significantly simplify the original proof, which might appear to be somewhat {ad hoc}. Instead, our arguably simpler achievability proof (Theorem~\ref{thm:CSCCH}) reveals a structure of the composition of known protocols.  We can achieve this due to the observation that connects coded source compression with classical channel simulation, a powerful information-theoretical tool obtained recently \cite{Bennett:2002fw}. The proof indicates that channel simulation is a general subroutine employed between the helper and the decoder in the task of coded source compression. The approach of composing protocols is advocated under the name of ``resource inequalities'', and has been successful in quantum information theory \cite{Devetak:2008eb, Hsieh:2010gd, Datta:2011vc}.

Next, we consider classical source coding with a quantum helper.  In our setup, the quantum side information is observed by the helper, and  the decoder only has a classical description from the quantum helper.  Although our problem can be regarded as a classical-quantum version of the classical helper problem studied in \cite{Wyner:1975iv, Ahlswede:1975ea}, in contrast to its classical counterpart,  our problem does not reduce to source coding with quantum side information studied in \cite{Devetak:2003kd}, even if there is no constraint on rate $R_2$.  However, when the ensemble that constitutes the quantum side information is commutative, our problem reduces to the classical helper problem \cite{Wyner:1975iv, Ahlswede:1975ea}. We completely characterize the rate region of the classical source coding with the quantum helper problem. In fact, the formulae describing the rate region (cf.~Theorem \ref{Theorem:rate-region}) resembles the classical counterpart  (cf.~\eqref{eq:classical-region-R1} and \eqref{eq:classical-region-R2}).  However, the proof technique is very different due to the quantum nature of the helper.   In particular, we use Winter's measurement compression theory \cite{Winter:2004uk} in the direct coding theorem. One of the interesting consequences of our result is that a helper's scheme that separately conducts a measurement and a compression is suboptimal; measurement compression appears to be necessary to achieve the optimal rate region.

We then extend the classical distributed source coding problem \cite{Wyner:1975iv, Ahlswede:1975ea} and its classical-quantum generalisation to the fully quantum version; namely compression of a quantum source with the help of a quantum server. We consider two natural scenarios depending on the types of communication channels, be it classical or quantum, shared between the legitimate sender and the receiver. Moreover, we allow the quantum helper to utilize the quantum resource of entanglement freely. This assumption could somehow be justified since the centralized server normally has more resources at his disposal. Our direct coding proofs compose two fundamental quantum protocols; the state merging protocol \cite{Horodecki:2005fv, Horodecki:2006hl} and the fully quantum Slepian-Wolf protocol, together with the quantum reverse Shannon theorem  \cite{Bennett:2014il}. 
The current progress of coded source compression is summarized in Table~\ref{table:main}. We hope that this observation will lead to concrete progress of the \emph{quantum} multi-terminal noiseless source coding problem. 
\begin{table}[ht!]
\label{table:main}
\begin{center}
\begin{tabular}{|c|c|c|}\hline
& {Classical Source} & Quantum Source\\ \hline
\multirow{1}{*}{Classical Helper} & {Refs.~\cite{Wyner:1975iv, Ahlswede:1975ea}} & $\times$  \\ \hline
\multirow{1}{*}{Quantum Helper} & {Theorem~\ref{Theorem:rate-region}} & \multirow{1}{*}{Theorem~\ref{thm_main}}  \\ \hline
\end{tabular}
\end{center}
\vspace{2mm}
\caption{Coded Source Coding in Classical and Quantum Regimes}
\end{table}%
 We remark that the quantum source compression with a classical helper is a very subtle task, which is left open even when the decoder has full classical side information \cite{Devetak:2003kd} (instead of partial side information from the classical helper).

The quantum source compression with a quantum helper is treated in a different scenario in Ref.~\cite{Horodecki:2006hl}. Over there, classical communication is allowed from the helper to the receiver, and limited entanglement resource is considered. As a result, their formula requires regularization. By contrast, our result resorts to a quantum reverse Shannon theorem, and has the appealing single-letter expression.

There is a huge amount of work devoted to both classical and quantum lossy source coding \cite{Shannon:1959tf, Berger71, Devetak:2002it, Datta:2013ur, Wilde:2013hp, Datta:2013jk}. We will restrict ourselves to only noiseless source coding in this work.  However, as it turns out, channel simulation simplifies both rate distortion theory and coded source compression. 

\emph{Notations.} In this paper, we will use capital letters $X,Y,Z,U$ etc.~to denote classical random variables, and lower cases $x,y,z,u$  to denote their realisations. We use $\cX,\cY,\cZ,\cU$ to denote the sample spaces. We denote $x^n=x_1x_2\cdots x_n$. 

A quantum state is a positive semi-definite matrix with trace equal to one. We will use $\rho$ or $\sigma$ to denote quantum states in this paper. Furthermore, we will reserve the notation $\tau$ to denote the completely mixed state. In case we need to specify which party the quantum state belongs to, we will use a subscript description $\rho_A$, meaning that the quantum system is held by A(lice). Letting $\{\proj{x}\}_{x\in\cX}$ be a set of orthonormal basis vectors, a classical-quantum state $\rho_{XB}$ is written as
\begin{align*}
\rho_{XB}=\sum_{x} p_{X}(x) \proj{x} \otimes \rho_{x},
\end{align*}
so that $n$ copies of it is
\begin{align*}
\rho_{XB}^{\otimes n}=\sum_{x^n} p_{X}^{(n)}(x^n) \proj{x^n} \otimes \rho_{x^n},
\end{align*}
where we denote $\rho_{x^n}:=\rho_{x_1}\otimes\cdots\otimes\rho_{x_n}$ for the sequence $x^n$. A positive-operator valued measure (POVM), $\Lambda=\{\Lambda_y\}_{y\in\cY}$, is a quantum measurement whose elements are non-negative self-adjoint operators on a Hilbert space so that $\sum_{y\in\cY}\Lambda_y = I$. For a quantum channel $\cE:A\to B$, we will denote its Stinespring extension by $U_{\cE}:A\to BC$.

Various entropic quantities will be used in the paper. The von Neumann entropy of a quantum state $\rho_A$, where the subscript $A$ represents the quantum state held by A(lice), is $H(A)_\rho=-\tr(\rho_A\log\rho_A) $. The conditional von Neumann entropy of system $A$ conditioned on $B$ of a bipartite state $\rho_{AB}$ is $H(A|B)_\rho = H(AB)_\rho - H(B)_\rho$. The quantum mutual information between two systems $A$ and $B$ of $\rho_{AB}$ is $I(A;B)_\rho= H(A)_\rho+H(B)_\rho-H(AB)_\rho$. The conditional quantum mutual information $I(A;B|C)_\rho = I(A;BC)_\rho-I(A;C)_\rho$.

We will also employ the framework of \textit{Resource Inequality (RI)} \cite{Devetak:2008eb, Datta:2011vc}. The RIs are a concise way of describing interconversion of  resources in an information-processing task. Denote by $[qq]$ and $[q\to q]$ an ebit (maximally entangled pairs of qubits) and a noiseless qubit channel, respectively. Then a quantum channel $\cN$ that can faithfully transmit $Q(\cN)$ qubits per channel use with an unlimited amount of entanglement assistance can be  symbolically represented as
\[
\langle \cN \rangle +\infty [qq] \geq Q(\cN) [q\to q],
\]
where $\langle \cN \rangle$ is an asymptotic noisy resource that corresponds to many independent uses, i.e. $\cN^{\otimes n}$. Schumacher's noiseless source compression \cite{Schumacher:1995dg} can be similarly expressed
\[
H(B)_\rho [q\to q] \geq  \langle \rho_B\rangle,
\] 
which means that a rate of $H(B)_\rho$ noiseless qubits \emph{asymptotically} is sufficient to represent the noisy quantum source $\rho_B$. 

Sometimes, the RI only applies to the relative resource, $\langle \cN:\rho \rangle$, which means that the asymptotic accuracy is achieved only when $n$ uses of $\cN$ are fed an input of the form $\rho^{\otimes n}$. For detailed treatment of combining two RIs and cancellations of quantum resources, see Ref.~\cite{Devetak:2008eb}.

This paper is organised as follows. In Sec~\ref{secI}, we revisit classical source compression with a helper using the channel simulation idea.  In Sec~\ref{secII}, we formally define the problem of source coding with a quantum helper, and present the main result as well as its proof. In Sec~\ref{sec_full_quantum}, we study the source coding with a helper in the fully quantum regime, and two different scenarios will be treated. We conclude in Sec~\ref{secIII} with open questions.

\section{Classical Coded Source Compression}\label{secI}

We first define the classical channel simulation capacity, paralleling the definition in \cite[Eq. (4)]{Bennett:2002fw}. Consider two classical channels $W_1:\mathcal{X}_1\to\mathcal{Y}_1$ and $W_2: \mathcal{X}_2\to\mathcal{Y}_2$. 
To be completely general, we also allow (sufficiently many) common randomness
$\overline{\Phi}_{AB}$ on $\mathcal{A} \times \mathcal{B}$ shared between Alice's encoding and Bob's decoding operations, where $\mathcal{A}=\mathcal{B}$.
For given integers $n$ and $m$, a channel simulation code is specified by 
\begin{itemize}
\item Alice's encoding operation $E_n: \mathcal{X}_2^m \times \mathcal{A} \to \mathcal{X}_1^n$ and 
\item Bob's decoding operation $D_n:\mathcal{Y}_1^n\times \mathcal{B} \to\mathcal{Y}_2^m$.
\end{itemize}
We say that a given channel simulation code is $(n,m,\epsilon)$-code if, for every $x^m\in \mathcal{X}_2^m$,
\[
\| W_2^{\otimes m}(x^m) - (I_A \otimes D_n) \circ ((W_1^{\otimes n}\circ E_n) \otimes I_B) (x^m\otimes \overline{\Phi}_{AB}) \|_1\leq \epsilon.
\]
Then, we say that rate $R$ is achievable if, for any $\epsilon,\delta>0$ and all sufficiently large $n$, there exists an $(n,m,\epsilon)$ code with $m \ge n(R-\delta)$.
The channel simulation capacity of the first channel $W_1$ needed to simulate the second channel $W_2$, $C(W_1\to W_2)$, is the supremum of all achievable rates.
Shannon's noiseless channel coding theorem can be seen as a special case of $C(W_1\to W_2)$ in which $W_2$ is the identity channel $\rm{id}$. 
Another special case is the following classical reverse Shannon Theorem where $W_1$ is the identity channel $\rm{id}$ (\cite[Theorem 2]{Bennett:2002fw} and \cite[Theorem 1(a)]{Bennett:2014il}).  
\begin{theorem}[Classical Reverse Shannon Theorem]  \label{thm:CRST}
\[
C({\rm id}\to N) = \frac{1}{C(N)},
\]
where $C(N)$ is the Shannon capacity of the channel $N$.
\end{theorem}

Using a noiseless resource to simulate a noisy one seems like a useless task to explore at first glance. However, it turns out that such a task can be used to simplify coded source compression conceptually (among others). 

Consider a discrete memoryless source $(X,Y)$ with a joint distribution $p_{XY}(x,y)$ over $\cX\times\cY$. To begin with, we define an $(n,\epsilon)$ code for classical coded source compression that consists of the following:
\begin{itemize}
\item Alice's encoding operation $\varphi_A: {\mathcal X}^n\to {\mathcal M}$, where ${\mathcal M}:=\{1,2,\cdots,|{\mathcal M}|\}$ and $|{\mathcal M}|=2^{nR_1}$;
\item Bob's encoding operation $\varphi_B: \cY^n \to\cL$, where $\cL:=\{1,2,\cdots,|\cL|\}$ and $|\cL|=2^{nR_2}$;
\item a decoding operation $\cD:\cM\times\cL\to \widehat{\cX}^n$ 
\end{itemize}
so that
\[
\Pr\{X^n\neq \widehat{X}^n\}\leq \epsilon.
\]
A rate pair $(R_1,R_2)$ is said to be achievable if for any $\epsilon,\delta>0$ and all sufficiently large $n$, there exists an $(n,\epsilon)$ code with the rates $R_1+\delta$ and $R_2+\delta$. The rate region is defined as the collection of all achievable rate pairs.

\begin{theorem}[Classical source compression with a classical helper \cite{Wyner:1975iv, Ahlswede:1975ea}]\label{thm:CSCCH}
The optimal rate region for lossless source coding of $X$ with a classical helper $Y$ is the set of rate pairs $(R_1,R_2)$ such that 
\begin{align}
R_1 & \geq  H(X|U), \label{eq:classical-region-R1} \\
R_2 & \geq  I(U;Y), \label{eq:classical-region-R2} 
\end{align}
for some conditional distribution $p_{U|Y}(u|y)$, and $I(U;Y)$ is the classical mutual information between random variables $U$ and $Y$.
\end{theorem}
With the help of Theorem~\ref{thm:CRST}, we can provide a simpler direct coding theorem for Theorem~\ref{thm:CSCCH}.  

\begin{proof}
In the proof, the strategy that the classical helper employs can be conceptually viewed as assisting the decoder to simulate the local channel $p_{U|Y}$. For $n$ sufficiently large, the classical communication rate that the helper needs to send is $I(U;Y)$ and
\[
\| p_{UX}^n - q_{U_n X^n} \|_1\leq \epsilon,
\]
where $q_{U_n X^n}$ is the joint distribution induced by the simulation of $p_{U|Y}$.
Thus, the full side information about $U^n$ is possessed by the decoder, and the source compression with full side information can be carried out.  Since the helper's local channel can be simulated at the decoder whose inaccuracy is at most $\epsilon$, the overall error for classical source compression with a classical helper can be achieved with this additional $\epsilon$ error. 
\end{proof}

\section{Classical Source Compression with a Quantum Helper}\label{secII}
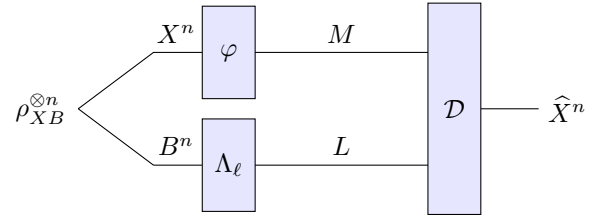
\begin{figure}
\centerline{
    \begin{tikzpicture}[scale=1][very thick]
    \fontsize{10pt}{1} 
    \tikzstyle{halfnode} = [draw,fill=white,shape= underline,minimum size=1.0em]
    \tikzstyle{checknode} = [draw,fill=blue!10,shape= rectangle,minimum height=3.5em, minimum width=2em]
    \tikzstyle{checknode2} = [draw,fill=blue!10,shape= rectangle,minimum height=8em, minimum width=2em]
    \tikzstyle{variablenode} = [draw,fill=white, shape=circle,minimum size=0.8em]
    \node (e1) at (-2.5,1) {$\rho_{XB}^{\otimes n}$} ;
    \node (p3) at (4.5,1) {$\widehat{X}^n$} ;
    \node (w2) at (1.5,2) {$M$} ;
    \node (l2) at (1.5,0.5) {$L$} ;
    \node (s2) at (-0.7,2) {$X^n$} ;
    \node (e2) at (-0.7,0.5) {$B^n$} ;
     \draw  (-2,1)-- (-1,1.75) (-2,1)-- (-1,0.25) (-1,0.25) -- (0,0.25) (-1,1.75) -- (0,1.75);
     \draw (3.3,1) --++ (p3);
     \draw (0.3,0.25) --(3,0.25) (0.3,1.75)--(3,1.75);
    \node[checknode] (cn1) at (0,1.75) {${\varphi}$};
    \node[checknode2] (cn2) at (3,1) {${\cal D}$};
    \node[checknode] (cn3) at (0,0.25) {${\Lambda}_{\ell}$};  
   \end{tikzpicture}
 }

  \caption{Classical Source Compression with a Quantum Helper.}\label{fig:QSCQH0}
\end{figure}

As shown in Figure~\ref{fig:QSCQH0}, the protocol for classical source coding with a quantum helper involves two senders, Alice and Bob, and one receiver, Charlie. Initially Alice and Bob hold $n$ copies of a classical-quantum state $\rho_{XB}$.
In this case, Alice holds classical random variables $X^n$ while Bob (as a helper) holds a quantum system $B^n$ that is correlated to Alice's message. The goal is for the decoder, Charlie, to faithfully recover Alice's message when assisted by the quantum helper, Bob.

We now proceed to formally define the coding procedure. We define an $(n,\epsilon)$ code for classical source compression with a quantum helper that consists of the following:
\begin{itemize}
\item Alice's encoding operation $\varphi: \cX^n \to \cM$, where $\cM:=\{1,2,\cdots, |\cM|\}$ and $|\cM|=2^{nR_1}$;
\item Bob's POVM $\Lambda=\{\Lambda_\ell\}: B^n \to \cL$, where $\cL:=\{1,2,\cdots,|\cL|\}$ and $|\cL|=2^{nR_2}$;
\item  Charlie's decoding operation $\cD:\cM\times\cL\to \widehat{\cX}^n$
\end{itemize}
so that the error probability satisfies
\begin{align}\label{eq_cond_cq}
\Pr\{X^n\neq \widehat{X}^n\}\leq \epsilon.
\end{align}

A rate pair $(R_1,R_2)$ is said to be \emph{achievable} if for any $\epsilon,\delta>0$ and all sufficiently large $n$, there exists an $(n,\epsilon)$ code with the rates $R_1+\delta$ and $R_2+\delta$. The rate region is then defined as the collection of all achievable rate pairs. Our main result is the following theorem. 

\begin{theorem} \label{Theorem:rate-region}
Given is a classical-quantum source $\rho_{XB}$. The optimal rate region for lossless source coding of $X$ with a quantum helper $B$ is the set of rate pairs $(R_1,R_2)$ such that 
\begin{align}
R_1 &\geq  H(X|U) \\
R_2 &\geq  I(U;B)_\sigma. 
\end{align}
The state $\sigma_{UB}(\Lambda)$ resulting from Bob's application of the POVM $\Lambda=\{\Lambda_u\}_{u\in\cU}$ is 
\begin{align}\label{eq_CQSTS}
\sigma_{UB}(\Lambda) = \sum_{u\in\cU} p_U(u) \proj{u}\otimes \rho_u
\end{align}
where
\begin{align}
p_U(u) &= \Tr(\rho_B\Lambda_u) \\
\rho_u &= \frac{1}{p_U(u)} [\sqrt{\rho_B} \Lambda_u \sqrt{\rho_B}]^* \\
\rho_B&= \sum_x p_X(x)\rho_x.
\end{align}
where $*$ denotes complex conjugation in the standard basis. Furthermore, we can restrict 
the size of POVM as $|\cU| \le d_B^2$, where $d_B$ is the dimension of Bob's system. 
\end{theorem}

A typical shape for the rate region in Theorem \ref{Theorem:rate-region} is illustrated in Figure~\ref{Fig:region}. When there is no constraint on $R_2$, rate $R_1$ can be decreased to be as small as
\begin{align}
H(X|U^*) &:= \min_{\Lambda} H(X|U) \\
&= H(X) - \max_{\Lambda} I(X; U) \\
&= H(X) - I_{\mathrm{acc}},
\end{align}
where $I_{\mathrm{acc}}$ is the accessible information for the ensemble $\{ p_X(x), \rho_x \}_{x\in{\cal X}}$.
Unless the ensemble commutes \cite{Holevo:1973wo},
the minimum rate $H(X|U^*)$ is larger than the rate $H(X|B)_\rho$, which is the optimal rate
in the source coding with quantum side information \cite{Devetak:2003kd}. 
To achieve $R_1 = H(X|U^*)$, it suffices to have $R_2 \ge I(U^*;B)_\sigma$, which is smaller than
$H(U^*)$ in general. This means that the following separation scheme is suboptimal:
first conduct a measurement to get $U^*$ and then compress $U^*$.
For more detail, see the direct coding proof.

\begin{figure}[t]
\centering{
\begin{minipage}{.4\textwidth}
\includegraphics[width=\textwidth]{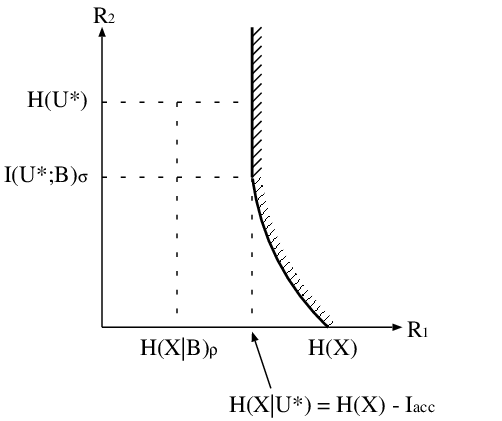}
\caption{A typical shape for the rate region in Theorem \ref{Theorem:rate-region}.}
\label{Fig:region}
\end{minipage}
}
\end{figure}

\subsubsection{Converse Proof}

Let $\varphi: {\cal X}^n \to {\cal M}$ be Alice's encoder, and let
$\{ \Lambda_{\ell} \}_{\ell \in {\cal L}}$ be Bob's measurement. 
Alice sends $M = \varphi(X^n)$ to the decoder, and Bob sends the measurement outcome $L$ to the decoder. Fano's inequality states that $H(X^n|M,L)\leq n\epsilon_n$ for some $\epsilon_n\to 0$ as $n\to \infty$, given the condition (\ref{eq_cond_cq}) holds. 

First, we have the following bound:
\begin{align}
\log |{\cal M}| 
&\ge H(M) \\
&\ge H(M|L) \\
&\ge H(M|L) - H(M|L,X^n) \\
&= I (X^n;M|L) \notag\\
&= H(X^n|L) - H(X^n|ML) \\
&\stackrel{(a)}\ge H(X^n|L) - n \epsilon_n \\
&\stackrel{(b)}\ge \sum_{t=1}^n H(X_t|X_{<t},L) - n \epsilon_n \\
& \stackrel{(c)}= \sum_{t=1}^n H(X_t|U_t) - n \epsilon_n \\
& \stackrel{(d)}= n H (X_J|U_J, J) - n \epsilon_n,
\end{align}
where $(a)$ follows from Fano's inequality: $H(X^n|M,L)\leq n \epsilon_n$ for some $\epsilon_n\to 0$ as $n \to \infty$; (b) uses the chain rule and denotes $X_{<t} := (X_1,\ldots,X_{t-1})$; in $(c)$, we denote $U_t:=(X_{<t},L)$; $(d)$ introduces a time-sharing random variable $J$ that is uniformly distributed in the set $\{1,2,\cdots n\}$.

Denote the state after the helper's measurement $\{\Lambda_\ell\}$ by
\[
\omega_{LX^nB^n} : = \sum_{\ell} p_{L|X^n}(\ell|x^n)p_{X^n}(x^n) \proj{\ell} \otimes\proj{x^n} \otimes \rho_{\ell x^n}^{B^n},
\]
where 
\begin{eqnarray*}p_{L|X^n} (\ell|x^n)&=& \tr [\rho_{x^n}\Lambda_\ell] \\
\rho_{\ell x^n}^{B^n}&=& \frac{1}{p_{L|X^n}(\ell |x^n)}  [\sqrt{\rho_{x^n}}\Lambda_\ell \sqrt{\rho_{x^n}}]^*.
\end{eqnarray*} 
The following equations are evaluated on the state $\omega_{LX^nB^n}$: 
\begin{align}
\log |{\cal L}|
&\ge H(L) \\
&\ge I(L ; B^n) \\
&= \sum_{t=1}^n I(L; B_t|B_{<t}) \\
&= \sum_{t=1}^n I(L, B_{<t}; B_t) \\
&\stackrel{(a)}{=} \sum_{t=1}^n I(L, B_{<t}, X_{<t}; B_t) \\
&\ge \sum_{t=1}^n I(L, X_{<t}; B_t) \\
& = \sum_{t=1}^n I(U_t; B_t) \label{eq_L00}
\end{align}
where $(a)$ follows from 
\begin{align}
I(X_{<t} ; B_t|L,B_{<t}) 
&\le I(X_{<t}; B_t,B_{>t}|L,B_{<t}) \\
&= H(X_{<t}|L,B_{<t}) - H(X_{<t}|L,B^n) \\
&\le H(X_{<t}|B_{<t}) - H(X_{<t}|L,B^n)\\
&= H(X_{<t}|B_{<t}) - H(X_{<t}|B^n)\label{eq_L01}  \\
&= H(X_{<t}|B_{<t}) - H(X_{<t}|B_{<t}) \\
&= 0.
\end{align}
The equality in Eq.~(\ref{eq_L01}) follows since $L-B^n -X^n$ is a Markov chain.

Following from Eq.~(\ref{eq_L00}), we can again introduce a time-sharing random variable $J$ that is uniformly distributed in the set $\{1,2,\cdots, n\}$: 
\begin{align}
\sum_{t=1}^n I(U_t; B_t) & = n \sum_{t=1}^n  I(U_t; B_t |J=t) \\
&= n I(U_J; B_J|J) \\
&= n I(U_J J;B_J)
\end{align}
where the last equality follows because $I(J;B_J)=0$. To get the single-letter formula, define $X:=X_J$, $B:=B_J$, $U:=(U_J,J)$ and let $n\to \infty$:
\begin{align}
R_1 &= \frac{1}{n} \log |\cM| \geq H(X|U) \\
R_2 &=\frac{1}{n} \log|\cL| \geq I(U;B).
\end{align}
Note that the distribution of $U_t = (L,X_{<t})$ can be written as
\begin{align}
p_{X_{<t} L}(x_{<t},\ell) &=  \left(\prod_{i < t} p_X(x_i)\right) \times \notag \\
 & \Tr\left[ \left\{ \left( \bigotimes_{i < t} \rho_{x_i} \right) \otimes \rho_{B_t}
  \otimes \left( \bigotimes_{i > t} \rho_{B_i} \right)\right\}  \Lambda_\ell  \right],
\end{align}
where $\rho_{B_j} =\sum_{x_j\in\cX} p_X(x_j) \rho_{x_j}$. Thus, we can get $U_t$ as a measurement outcome of $B_t$ by first generating $X_{<t}$, then
by appending $\bigotimes_{i < t} \rho_{x_i}$ and $\bigotimes_{i > t} \rho_{x_i}$ to ancillae systems, 
and finally by conducting the measurement $\{ \Lambda_\ell \}_{\ell \in {\cal L}}$.

Finally, the bound on $|{\cal U}|$ can be proved via Carath\'odory's theorem (cf.~\cite[Appendix C]{Devetak:2005ea}).

\subsubsection{Direct Coding Theorem}
Fix a POVM $\Lambda=\{\Lambda_u\}_{u\in\cU}$. It induces a conditional probability $p_{U|X}(u|x)= \Tr [\Lambda_u\rho_x]$, and a joint probability distribution 
\begin{equation}\label{eq_Pau}
P_{XU}(x,u)=p_X(x)p_{U|X}(u|x).
\end{equation} 

The observation made in the achievability proof is the application of Winter's measurement compression theory \cite{Winter:2004uk}. 

\begin{theorem}[Measurement compression theorem \cite{Winter:2004uk, Wilde:2012iq}] \label{thm:meas-comp}Let $\rho_A$ be a source state and $\Lambda$ a POVM\ to simulate on this state. A protocol for a
faithful simulation of the POVM\ is achievable with classical communication rate $R$ and common randomness rate $S$ if and only if the following set of inequalities hold%
\begin{align}
R   \geq I\left(  X;R\right) ,~~
R+S   \geq H\left(  X\right)  ,
\end{align}
where the entropies are with respect to a state of the following form:%
\begin{equation}
\sum_{x}\left\vert x\right\rangle \left\langle x\right\vert_{X}
\otimes \text{Tr}_{A}\left\{  \left(  I_{R}\otimes\Lambda_{x}
\right)  \phi_{RA}\right\}  ,\label{eq:IC-state}%
\end{equation}
and $\ket{\phi_{RA}}$ is some purification of the state $\rho_A$.
\end{theorem}

Let $K$ be a random variable on ${\cal K}$, which describes the common randomness shared between
Alice and Bob. Let 
$\{ \widetilde{\Lambda}_{u^n}^{(k)} \}_{u^n \in {\cal U}^n}$ be collection of POVMs. Let 
\begin{align}
Q_{X\widetilde{U}}^n(x^n,u^n) := P_X^{(n)}(x^n) \sum_{k \in {\cal K}} \frac{1}{|{\cal K}|} {\rm{Tr}}[ \rho_{x^n} \widetilde{\Lambda}_{u^n}^{(k)}],
\end{align}
where $P_X^{(n)}(x^n):=P_X(x_1)\times\cdots\times P_X(x_n)$.
The faithful simulation of $n$ copies of POVM $\Lambda:=\{ \Lambda_u \}_{u\in {\cal U}}$, i.e.~$\Lambda^{\otimes n}$, implies that for any $\epsilon>0$ and $n$ sufficiently large, there exist POVMs $\{\widetilde{\Lambda}^{(k)}\}_{k\in\cK}$, where $\widetilde{\Lambda}^{(k)}:=\{\widetilde{\Lambda}_{u^n}^{(k)} \}_{u^n \in {\cal U}^n}$, with  
\begin{align} \label{eq:faithful-simulation}
\frac{1}{2} \| P_{XU}^{(n)} - Q_{X\widetilde{U}}^{n} \|_1 \le \epsilon. 
\end{align}
Note that Winter's measurement compression theorem can be expressed in terms of a resource inequality (RI) as follows:
\[
I(X;R)[c\to c] + H(X|R) [cc] \geq \langle \Lambda: \rho_A \rangle, 
\] 
where the entropic quantities are defined with respect to the state in Eq.~(\ref{eq:IC-state}).


Now, we are ready to prove the direct coding theorem. The key insight is the equivalence between the role of a quantum helper in the coded source compression and the simulation of a quantum-classical channel (i.e.~a quantum measurement).
 
\emph{Bob's coding.} Recall the state $\sigma_{UB}(\Lambda)$ in Eq.~(\ref{eq_CQSTS}), where a classical random variable $U$ is generated after Bob's measurement $\Lambda=\{\Lambda_u\}_{u\in\cU}$ acting on his quantum system $B$. Note that Alice and Bob now hold classical random variables $X^n$ and $U^n$, respectively. Instead of performing the purely classical coding strategy on $X^nU^n$ stated in Theorem~\ref{thm:CSCCH}, the decoder Charlie can directly simulate the measurement outcome $U^n$ using Winter's  measurement compression theorem  \cite{Winter:2004uk, Wilde:2012iq}. 
Then Theorem~\ref{thm:meas-comp} promises that by sending $R_2\geq I(U;B)_\sigma$ from Bob to the decoder Charlie, Charlie will have a local copy $\widetilde{U}^{ n}$ that is $\epsilon$-close to $U^n$. Furthermore, the distribution $Q^n_{X\widetilde{U}}$ between Alice and Charlie  will satisfy Eq.~(\ref{eq:faithful-simulation}). 

\emph{Alice's coding.}  Alice's strategy is very simple once Charlie has had $\widetilde{U}^n$. She just uses the Slepian-Wolf coding as if she starts with the distribution $P_{XU}$ with Charlie. In fact, it is well known (cf.~\cite{Cover}) that an encoder $\varphi:\cX^n \to \cM$ and a decoder $\cD:\cM \times \cU^n \to \cX^n$ exist such that $|\cM| = 2^{n(H(X|U)+\delta)}$ and 
\begin{align}
P^{(n)}_{XU}({\cal A}^c) \le \epsilon
\end{align}
for sufficiently large $n$, where 
\begin{align}
{\cal A} := \{ (x^n,u^n) \in \cX^n \times \cU^n: \cD(\varphi(x^n),u^n) = x^n \}
\end{align}
is the set of correctably decodable pairs.

Now, suppose that Alice and Bob use the same code for the simulated distribution $Q_{X\widetilde{U}}^n$. Then, by the definition of the variational distance and Eq.~\eqref{eq:faithful-simulation}, we have
\begin{align}
 Q^n_{X\tilde{U}}({\cal A}^c) \le P^{(n)}_{XU}({\cal A}^c) + \epsilon.
\end{align}
Thus, if we can find a good code for $P_{XU}^{(n)}$, we can also use that code for $Q^n_{X\widetilde{U}}$ for a sufficiently large $n$.

\emph{Derandomization.} The standard derandomization technique works here. If the random coding strategy works fine on average, then there is one realisation that also works. Since the distribution $Q^{n}_{X\widetilde{U}} = \frac{1}{|\cK|} \sum_{k\in\cK}Q^n_{X\widetilde{U}|k}$, we have
\begin{align}
\sum_k \frac{1}{|{\cal K}|} Q^n_{X\widetilde{U}|K=k}({\cal A}^c) = 
 Q^n_{X\tilde{U}}({\cal A}^c) \le P_{XU}({\cal A}^c) + \epsilon.
\end{align}
Thus, there exists one $k\in\cK$ so that $Q^n_{X\widetilde{U}|k}(\cA^c)$ is small. 


\section{Fully Quantum Source Compression with a Quantum Helper}
\label{sec_full_quantum}

The coding scenario for quantum sources with a quantum helper is potentially more complicated than purely classical settings due to the existence of entanglement. We will discuss two scenarios that depend on the types of communication channels between the legitimate sender and the receiver. In both settings, we assume that the helper (or the centralized server) possesses quantum resources of entanglement shared between him and the receiver. Before showing the main results in this section, we first present relevant quantum protocols that prove crucial in the proof of achievability. 

\textit{Relevant quantum protocols.} Given a bipartite state $\rho_{AB}$ whose purification is $\ket{\psi_{ABR}}$, the state merging protocol \cite{Horodecki:2005fv, Horodecki:2006hl, Dupuis:2014jz} is the information-processing task of distributing the $A$-part of the system, that originally belongs to Alice, to the distant Bob without altering the joint state. Moreover, Alice and Bob have access to pre-shared entanglement and their goal is to minimise the number of EPR pairs consumed during the protocol. The state merging can be efficiently expressed as the following RI:
\begin{equation}\label{eq_SM}
\langle \psi_{A|B|R}\rangle +I(A;R)_\psi[c\to c] +H(A|B)_\psi[qq] \geq \langle \psi_{|AB|R}\rangle
\end{equation}
where the notation $\psi_{A|B|R}$ denotes the state is originally shared between three distant parties Alice, Bob, and Eve, while $\psi_{|AB|R}$ means that the system $A$ is now together with the system $B$. This protocol involves classical communication; however, for the purpose of this section, quantum resources are much more valuable and classical communication is considered to be free. As a result, the state merging protocol either consumes EPR pairs with rate $H(A|B)_\psi$ when this quantity is positive, or generates EPR pairs with rate $|H(A|B)_\psi|$ for later uses, if $H(A|B)_\psi$ is negative, after the transmission of the system $A$ to $B$. 

The state merging protocol gives the first operational interpretation to the conditional von Neumann entropy. More importantly, it provides an answer to a formerly long-standing puzzle---the conditional von Neumann entropy can be negative, a situation that has no classical correspondence.

The fully quantum Slepian-Wolf (FQSW) protocol \cite{Abeyesinghe:2009ej, Datta:2011vc} can be considered as the coherent version of the state merging protocol. It can be described as
\begin{equation}\label{eq_FQSW}
\langle \psi_{A|B|R}\rangle + \frac{1}{2} I(A;R)_\psi [q\to q] \geq \frac{1}{2} I(A;B)_\psi [qq] +\langle \psi_{|AB|R}\rangle.
\end{equation} 
It is a simple exercise to show, via the resource inequalities, that the state merging protocol can be obtained by combining teleportation with the FQSW 
protocol \cite{Abeyesinghe:2009ej, Datta:2011vc}. Moreover, the FQSW protocol can be transformed into a version of the quantum reverse Shannon theorem (QRST) that involves entanglement assistance \cite{Abeyesinghe:2009ej}. 

The quantum reverse Shannon theorem (QRST) addresses a fundamental task that asks, given a quantum channel $\cN$, how much quantum communication is required from Alice to Bob so that the channel $\cN$ can be simulated. There are variants of the QRST depending on whether entanglement or feedback is allowed in the simulation (see \cite[Theorem 3]{Bennett:2014il}).  The QRST protocol has become a powerful tool in quantum information theory. It can be used to establish a strong converse to the entanglement-assisted capacity theorem. Moreover, it can also be used to establish quantum rate distortion theorems  \cite{Datta:2013ur, Wilde:2013hp, Datta:2013jk}.

In this paper, we will use the QRST with entanglement assistance  \cite[Theorem 3(a)]{Bennett:2014il}. 
\begin{theorem}[Quantum Reverse Shannon Theorem]\label{thm_QRST}
 Let $\cN$ be a quantum channel from $A$ to $B$ so that its isometry $U_\cN:{A\to BE}$ results in the following tripartite state when inputting $\rho_A$:
\[
\ket{\psi_{RBE}}=I_R\otimes U_\cN \ket{\psi^\rho_{RA}},\]
where $\tr_R\proj{\psi_{RA}^\rho}=\rho_A$. Then with a sufficient amount of pre-shared entanglement, the channel $\cN$ with input $\rho_A$ can be simulated with quantum communication rate $\frac{1}{2}I(R;B)_\psi$:
\begin{equation}\label{eq:QRST}
\frac{1}{2} I(R;B)_\psi [q\to q]+ \frac{1}{2} I(E;B)_\psi[qq]\geq \langle \cN:\rho_A\rangle. 
\end{equation} 

\end{theorem}

\subsection{Scenario I: Classical Communication Channel between the Quantum Source and Receiver}

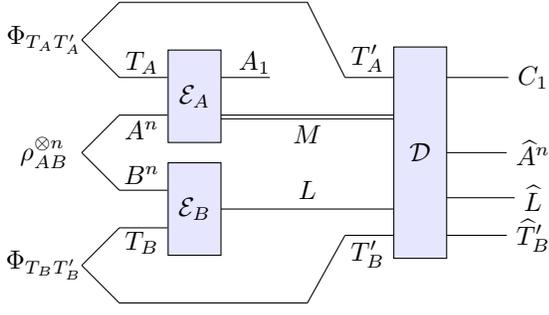
\begin{figure}
\centerline{
    \begin{tikzpicture}[scale=1][very thick]
    \fontsize{10pt}{1} 
    \tikzstyle{halfnode} = [draw,fill=white,shape= underline,minimum size=1.0em]
    \tikzstyle{checknode} = [draw,fill=blue!10,shape= rectangle,minimum height=3.5em, minimum width=2em]
    \tikzstyle{checknode2} = [draw,fill=blue!10,shape= rectangle,minimum height=8em, minimum width=2em]
    \tikzstyle{variablenode} = [draw,fill=white, shape=circle,minimum size=0.8em]
    \node (e1) at (-2,1) {$\rho_{AB}^{\otimes n}$} ;
    \node (EPR1) at (-2,2.5) {$\Phi_{T_AT_A'}$} ;
    \node (EPR2) at (-2,-0.5) {$\Phi_{T_BT_B'}$} ;
    \node (p3) at (4.5,1) {$\widehat{A}^n$} ;
    \node (p4) at (4.5,2.0) {$C_1$} ;
    \node (p5) at (4.5,0.4) {$\widehat{L}$} ;
    \node (p6) at (4.5,-0.1) {$\widehat{T}_{B}'$} ;
    \node (w2) at (1.5,1.25) {$M$} ;
    \node (w3) at (0.8,2.2) {$A_1$} ;
    \node (l2) at (1.5,0.5) {$L$} ;
    \node (s2) at (-0.7,1.3) {$A^n$} ;
    \node (e2) at (-0.7,0.7) {$B^n$} ;
    \node (s3) at (-0.7,2.2) {$T_A$} ;
    \node (e3) at (-0.7,-0.2) {$T_B$} ;
    \node (s4) at (2.3,2.25) {$T_A'$} ;
    \node (e4) at (2.3,-0.35) {$T_B'$} ;
     \draw  (-1.5,1)-- (-1,1.5) (-1.5,1)-- (-1,0.5) (-1,0.5) -- (0,0.5) (-1,1.5) -- (0,1.5); 
     \draw  (-1.5,-0.5)-- (-1,0) (-1.5,-0.5)-- (-1,-1.0) (-1,0) -- (0,0) (-1,-1) -- (1.5,-1) (1.5,-1)--(2,-0.1); 
     \draw  (-1.5,2.5)-- (-1,2) (-1.5,2.5)-- (-1,3.0) (-1,2.0) -- (0,2.0) (-1,3.0) -- (1.5,3.0) (1.5,3.0)-- (2, 2.0); 
     \draw (3.3,1) --++ (p3) (3.3,2.0) --++ (p4) (3.3,0.4) --++ (p5) (3.3,-0.1) --++ (p6); 
     \draw (0.3,0.25) --(3,0.25) (2,-0.1) --(3,-0.1) ; 
     \draw  (0.3,1.5)--(3,1.5) (0.3,1.45)--(3,1.45) (0.3, 2.0) -- (1,2.0) (2, 2.0) -- (3,2.0); 
    \node[checknode] (cn1) at (0,1.75) {${\cal E}_{A}$};
    \node[checknode2] (cn2) at (3,1) {${\cal D}$};
    \node[checknode] (cn3) at (0,0.25) {${\cal E}_{B}$};  
   \end{tikzpicture}
 }

  \caption{
Fully quantum source compression with a quantum helper when the quantum source and the receiver are connected by a classical channel.
  }\label{fig:QSCQH}
\end{figure}

As shown in Figure~\ref{fig:QSCQH}, the protocol for fully quantum source coding with a quantum helper involves two senders, Alice and Bob, and one receiver, Charlie. Initially Alice and Bob hold $n$ copies of a  bipartite quantum state $\rho_{AB}$, where Alice holds quantum systems $A^n:=A_1\cdots A_n$ while Bob (being a quantum helper) holds quantum systems ${B^n}=B_1\cdots B_n$. Moreover, there are pre-shared entangled states $\ket{\Phi_{T_AT_A'}} $ between Alice and Charlie, and pre-shared entangled states $\ket{\Phi_{T_BT_B'}} $ between Bob and Charlie. The goal is for the decoder Charlie to faithfully recover Alice's quantum state $\rho_{A^n}=\tr \rho_{AB}^{\otimes n}$ when assisted by the quantum helper Bob.

We now proceed to formally define the coding procedure. 
Let $\psi_{ABR}$ be a purification of $\rho_{AB}$.
We define an $(n,\epsilon)$ code for fully quantum source compression with a quantum helper to consist of the following:
\begin{itemize}
\item Alice's encoding operation $\cE_A: T_A A^n \to A_1 M$, where $A_1$ is a quantum system and $M$ is a classical system; Alice only sends $M$ to Charlie;
\item Bob's encoding operation $\cE_B: T_BB^n \to L$,  where $L$ is a quantum system of dimension $|L|=2^{nR_2}$; Bob sends the quantum system to Charlie;
\item  Charlie's decoding operation $\cD:M LT_{A}' T_B'\to \widehat{A}^n \widehat{L} \widehat{T}_B^\prime$ with isometric extension $U_\cD:M LT_{A}' T_B'\to C_1\widehat{A}^n \widehat{L} \widehat{T}_B^\prime$ that produces 
\[
\omega_{A_1C_1\widehat{A}^n \widehat{L}R^n \widehat{T}_B^\prime}=I_{A_1R^n}\otimes U_\cD(\sigma_{A_1MLR^nT_A'T_B'})
\]
where
\[
\sigma_{A_1MLR^nT_A'T_B'} = \cE_A\otimes I_{R^nLT_B'}(\theta_{A^nR^nLT_B'} \otimes \Phi_{T_AT_A'})
\]
and 
\[
\theta_{A^nR^nLT_B'} = I_{A^nR^nT_B'}\otimes {\cE_B} (\psi_{ABR}^{\otimes n} \otimes \Phi_{T_BT_B'}).
\]
\end{itemize}
We demand the final state to satisfy
\begin{align}\label{eq_cond1}
\| \omega_{A_1C_1\widehat{A}^n R^n\widehat{L} \widehat{T}_B^\prime} - \Phi_{A_1C_1}\otimes\theta_{{A}^n R^nL  T_B^\prime}\|_1\leq \epsilon,
\end{align}
where $\ket{\Phi_{A_1C_1}}$ is a maximally entangled state. This condition (\ref{eq_cond1}) guarantees that Alice's quantum system $A^n$ is sent to Charlie faithfully since $\omega_{\widehat{A}^n}\approx_\epsilon\theta_{A^n}=\rho_{A^n}$.  

Let $R_1=\log|T_A|-\log|A_1|$. 
A rate pair $(R_1,R_2)$ is said to be \emph{achievable} if for any $\epsilon,\delta>0$ and all sufficiently large $n$, there exists an $(n,\epsilon)$ code with rates $R_1+\delta$ and $R_2+\delta$. The rate region is then defined as the collection of all achievable rate pairs. Our main result is the following theorem. 

\begin{theorem}\label{thm_main}
Given is a bipartite quantum source $\rho_{AB}=\tr_R \psi_{ABR}$. The optimal rate region for lossless source coding of $A$ with a quantum helper $B$ is the set of rate pairs $(R_1,R_2)$ such that 
\begin{eqnarray}
R_1 &\geq & H(A|C)_\phi \label{thm6_1} \\
R_2 &\geq &\frac{1}{2} I(RA;C)_\phi. 
\end{eqnarray}
The state $\phi_{ACER}$ resulting from Bob's application of some CPTP map $\cE:{B\to C}$ with isometric extension   $U_\cE:{B\to CE}$ is 
\begin{equation}\label{eq_state01}
\ket{\phi_{ACER}}= I_{RA}\otimes U_\cE\ket{\psi_{ABR}}.
\end{equation}
\end{theorem}

\subsubsection{Direct part} 


We use the channel simulation method. 
For any local channel ${\cal E}:{B \to C}$ performed by the quantum helper $B$ on his half of bipartite state $\rho_{AB}$, it can be simulated by the decoder using the quantum reverse Shannon theorem (QRST) (Theorem~\ref{thm_QRST}):
\begin{align}
\frac{1}{2} I(RA; C)_\phi[q \to q] + \frac{1}{2} I(E; C)_\phi[qq] \ge \langle {\cal E}: \rho_B \rangle,
\end{align}
where $\ket{\phi_{ACER}}$ is given in Eq.~(\ref{eq_state01}).
In other words, by using the pre-shared entanglement between the helper and the decoder with rate $\frac{1}{2} I(E;C)_\phi$ and sending quantum message from the helper to the decoder with rate $\frac{1}{2} I(RA; C)_\phi$, the decoder can simulate the quantum state ${\cal E}(\rho_B)$ locally with error going to zero in the asymptotic sense.

Alice's coding: Once the decoder has the system $C$,  Alice and Charlie start the state merging protocol (\ref{eq_SM}), using the pre-shared entanglement with rate $H(A|C)_\phi$.

\subsubsection{Converse part}

We refer to Figure~\ref{fig:QSCQH} for corresponding labels used in the converse proof. Recall the condition (\ref{eq_cond1}) states that the output state  $\omega_{A_1C_1\widehat{A}^n R^n \widehat{L} \widehat{T}_B^\prime} \approx_\epsilon \Phi_{A_1C_1}\otimes\theta_{{A}^n R^n L  T_B^\prime}$, where
\[
\theta_{A^nR^nLT_B'} = I_{A^nR^nT_B'}\otimes {\cE_B} (\psi_{ABR}^{\otimes n} \otimes \Phi_{T_BT_B'}).
\]
We will omit the state $\theta$ in the subscript in the following sequences of equations to simplify the notation. 

To bound $n R_1=\log|T_A|-\log|A_1|$, we  follow steps in the converse proof of the state merging protocol \cite{Horodecki:2006hl} and have
\begin{align}
n R_1 &\geq H(A^n|LT_B')_\theta + f(\epsilon)\\
&= \sum_{i=1}^n H(A_i |LT_B' A_{<i})_\theta + f(\epsilon)\\
&= \sum_{i=1}^n H(A_i | U_i)_\theta + f(\epsilon) \label{eq_qccoded01}\\
&= n H(A_T | U_TT)+ f(\epsilon)  \\
&= n H(A | C)+ f(\epsilon),
\end{align}  
where $f(\epsilon)\to 0$ as $\epsilon\to 0$. We set $U_i := (L,T_B', A_{<i})$ in Eq.~(\ref{eq_qccoded01}). We relabel $A:=A_T$ and $C:=(U_T,T)$ in the last equality, so that we recover Eq.~(\ref{thm6_1}). 

To bound the quantum communication rate $R_2=\log|L|$, we follow steps in the converse proof of the entanglement-assisted quantum rate-distortion theorem (see Eq.~(21) in \cite{Wilde:2013hp}):
\begin{align}
2n R_2 &\ge I(LT_B' ; R^n A^n)_\theta \label{eq_ffb} \\
&= \sum_{i=1}^n  I(LT_B'; R_iA_i|R_{<i}A_{<i})_\theta \\
&= \sum_{i=1}^n [ I(LT_B'R_{<i}A_{<i}; R_iA_i)_\theta- I(R_{<i}A_{<i} ; R_iA_i)_\theta] \\
&= \sum_{i=1}^n I(LT_B'R_{<i}A_{<i}; R_iA_i)_\theta \\
&\ge \sum_{i=1}^n I(LT_B' A_{<i}; R_iA_i)_\theta \\
&= \sum_{i=1}^n I(U_i ; R_iA_i)_\theta \label{eq_ui} \\
&= n I(U_T ; R_T A_T|T) \\
&= n I(U_TT; R_TA_T) \\
&= n I(C; RA). \label{eq_ffe}
\end{align}
Note that $U_i$ can be generated from $B_i$ via Bob's local CPTP, and $T$ again is a time-sharing random variable that is uniformly distributed in the set $\{1,2,\cdots, n\}$.
In fact, Bob can first append the maximally entangled states $(T_B,T_B')$, systems $(A_{<i},B_{<i})$, and $B_{>i}$. Then, he can perform ${\cal E}_B$, and get $U_i = (L,T_B', A_{<i})$. With the relabelling, it is clear that Eq.~(\ref{eq_ffe}) is evaluated on a quantum state of the form $\ket{\phi_{ACER}}$ in Eq.~(\ref{eq_state01}).


\subsection{Scenario II: Quantum Communication Channel between the Quantum Source and Receiver}

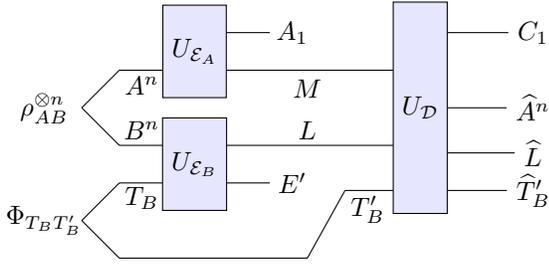
\begin{figure}
\centerline{
    \begin{tikzpicture}[scale=1][very thick]
    \fontsize{10pt}{1} 
    \tikzstyle{halfnode} = [draw,fill=white,shape= underline,minimum size=1.0em]
    \tikzstyle{checknode} = [draw,fill=blue!10,shape= rectangle,minimum height=3.5em, minimum width=2em]
    \tikzstyle{checknode2} = [draw,fill=blue!10,shape= rectangle,minimum height=8em, minimum width=2em]
    \tikzstyle{variablenode} = [draw,fill=white, shape=circle,minimum size=0.8em]
    \node (e1) at (-2,1) {$\rho_{AB}^{\otimes n}$} ;
    \node (EPR2) at (-2,-0.5) {$\Phi_{T_BT_B'}$} ;
    \node (p3) at (4.5,1) {$\widehat{A}^n$} ;
    \node (p4) at (4.5,2.0) {$C_1$} ;
    \node (p5) at (4.5,0.4) {$\widehat{L}$} ;
    \node (p6) at (4.5,-0.1) {$\widehat{T}_{B}'$} ;
    \node (w2) at (1.5,1.25) {$M$} ;
    \node (w3) at (1.3,2.0) {$A_1$} ;
    \node (w4) at (1.3,0) {$E'$} ;    
    \draw (0.3,0)--(1,0); 
    \node (l2) at (1.5,0.7) {$L$} ;
    \draw (0.3,0.5) --(3,0.5); 
    \node (s2) at (-0.7,1.3) {$A^n$} ;
    \node (e2) at (-0.7,0.7) {$B^n$} ;
   \node (e3) at (-0.7,-0.2) {$T_B$} ;
    \node (e4) at (2.3,-0.35) {$T_B'$} ;
      \draw  (-1.5,1)-- (-1,1.5) (-1.5,1)-- (-1,0.5) (-1,0.5) -- (0,0.5) (-1,1.5) -- (0,1.5); 
     \draw  (-1.5,-0.5)-- (-1,0) (-1.5,-0.5)-- (-1,-1.0) (-1,0) -- (0,0) (-1,-1) -- (1.5,-1) (1.5,-1)--(2,-0.1); 
     \draw (3.3,1) --++ (p3) (3.3,2.0) --++ (p4) (3.3,0.4) --++ (p5) (3.3,-0.1) --++ (p6); 
     \draw (2,-0.1) --(3,-0.1) ; 
     \draw  (0.3,1.5)--(3,1.5)  (0.3, 2.0) -- (1,2.0); 
    \node[checknode] (cn1) at (0,1.75) {$U_{{\cal E}_{A}}$};
    \node[checknode2] (cn2) at (3,1) {$U_{\cal D}$};
    \node[checknode] (cn3) at (0,0.25) {$U_{{\cal E}_{B}}$};  
   \end{tikzpicture}
 }

  \caption{
Fully quantum source compression with a quantum helper when the quantum source and the receiver are connected by a quantum channel.
  }\label{fig:QSCQH2}
\end{figure}

In this scenario, we replace the classical channel between Alice and Charlie with a quantum channel, as shown in Figure~\ref{fig:QSCQH2}. We define an $(n,\epsilon)$ code for fully quantum source compression with a quantum helper to consist of the following:
\begin{itemize}
\item Alice's encoding operation $\cE_A:A^n\to M$ with isometric extension  $U_{\cE_A}: A^n \to A_1 M$, where the quantum system $M$ to be sent is of size $|M|=2^{nR_1}$;
\item Bob's encoding operation $\cE_B:T_BB^n\to L$ with $U_{\cE_B}: T_BB^n \to L E'$, where the quantum system $L$ to be sent is of size $|L|=2^{nR_2}$;
\item  Charlie's decoding operation $\cD:M L  T_B'\to \widehat{A}^n \widehat{L} \widehat{T}_B^\prime$ with $U_{\cD}:M L  T_B'\to C_1\widehat{A}^n \widehat{L} \widehat{T}_B^\prime$ that produces 
\[
\ket{\omega_{A_1C_1E'\widehat{A}^n R^n\widehat{L}\widehat{T}_B^\prime}}=(I_{A_1R^nE'}\otimes U_\cD)\ket{\sigma_{A_1MR^nLE'T_B'}}
\]
where
\[
\ket{\sigma_{A_1MR^nLE'T_B'}} = (U_{\cE_A}\otimes I_{R^nLE'T_B'} )\ket{\theta_{A^nR^nLE'T_B'}}
\]
and 
\[
\ket{\theta_{A^nR^nLE'T_B'}} = I_{A^nR^nT_B'}\otimes U_{\cE_B} (\ket{\psi_{ABR}}^{\otimes n} \otimes \ket{\Phi_{T_BT_B'}}).
\]
\end{itemize}
We require that the final state satisfies
\begin{align}\label{eq_cond21}
\| \omega_{A_1C_1\widehat{A}^n R^n\widehat{L}\widehat{T}_B^\prime}- \Phi_{A_1C_1}\otimes\theta_{{A}^n R^nL T_B^\prime }\|_1\leq \epsilon,
\end{align}
where $\ket{\Phi_{A_1C_1}}$ is a maximally entangled state. The condition (\ref{eq_cond21}) guarantees that the output state $\omega_{\widehat{A}^nR^n\widehat{L}\widehat{T}_B^\prime}$ is close to the state $\theta_{A^nR^nLT_B'}$. Consequently, $\omega_{\widehat{A}^n}\approx_{\epsilon} \theta_{A^n}=\rho_{A^n}$.

A rate pair $(R_1,R_2)$ is said to be \emph{achievable} if for any $\epsilon,\delta>0$ and all sufficiently large $n$, there exists an $(n,\epsilon)$ code with the rates $R_1+\delta$ and $R_2+\delta$. The rate region is then defined as the collection of all achievable rate pairs. Our main result is the following theorem. 

\begin{theorem}\label{thm_main2}
Given a bipartite quantum source $\rho_{AB}=\tr_R \psi_{ABR}$, the optimal rate region for lossless source coding of $A$ with a quantum helper $B$ is the set of rate pairs $(R_1,R_2)$ such that 
\begin{eqnarray}
R_1 &\geq & \frac{1}{2} I(A;ER)_\phi \\
R_2 &\geq &\frac{1}{2} I(RA;C)_\phi. 
\end{eqnarray}
The state $\phi_{ACER}$ resulting from Bob's application of some CPTP map $\cE:{B\to C}$ with isometric extension $U_{\cE}:B\to CE$ is 
\begin{equation}\label{eq_state02}
\ket{\phi_{ACER}}= I_{RA}\otimes U_\cE \ket{\psi_{ABR}}.
\end{equation}
\end{theorem}

\subsubsection{Direct part}

Again, we use the channel simulation method. 
Any local channel ${\cal E}:{B \to C}$ performed by the quantum helper $B$ on his half of bipartite state $\rho_{AB}$ can be simulated by the decoder using the quantum reverse Shannon theorem (QRST) (Theorem~\ref{thm_QRST}):
\begin{align}
\frac{1}{2} I(RA; C)_\phi[q \to q] + \frac{1}{2} I(E; C)_\phi[qq] \ge \langle {\cal E}: \rho_B \rangle,
\end{align}
where $\phi_{ACER}$ is given in Eq.~(\ref{eq_state02}).
In other words, by using the pre-shared entanglement between the helper and the decoder with rate $\frac{1}{2} I(E;C)_\phi$ and sending quantum message from the helper to the decoder with rate $\frac{1}{2} I(RA; C)_\phi$, the decoder can simulate the quantum state ${\cal E}(\rho_B)$ locally with error goes to zero in the asymptotic sense.

Alice's coding: Once the decoder has the quantum system $C$, Alice starts the FQSW protocol (\ref{eq_FQSW}) with Charlie to merge her quantum system $A$ to $C$. The needed quantum communication  rate is $\frac{1}{2} I(A;ER)_\phi$.

\subsubsection{Converse part}

We refer to Figure~\ref{fig:QSCQH2} for corresponding labels used in the converse proof. 
Denote the states after Bob's encoding, Alice's and Bob's encodings, and all three operations as
\begin{eqnarray}
\theta_{A^nLE'T_B'} &=& I_{A^nT_B'}\otimes U_{\cE_B} (\rho_{AB}^{\otimes n} \otimes \Phi_{T_BT_B'})\\
\sigma_{A_1MLE'T_B'} &=& U_{\cE_A}\otimes I_{LE'T_B'} (\theta_{A^nLE'T_B'}) \\
\omega_{A_1C_1\widehat{A}^n\widehat{L}\widehat{T}_B' }&=& I_{A_1E'}\otimes U_\cD (\sigma_{A_1MLT_B'}),
\end{eqnarray}
and note that the reference system $R^n$ that purifies $\rho_{AB}^{\otimes n}$ also purifies $\theta_{A^nLE'T_B'} $, $\sigma_{A_1MLE'T_B'}$ and $\omega_{A_1C_1E'\widehat{A}^n\widehat{L}\widehat{T}_B' }$. Also recall that the state $\omega_{A_1C_1\widehat{A}^nR^n\widehat{L}\widehat{T}_B' }$ satisfies the condition (\ref{eq_cond21}); hence $\omega_{A_1C_1\widehat{A}^nR^n\widehat{L}\widehat{T}_B' }\approx_{\epsilon} \Phi_{A_1C_1}\otimes\theta_{{A}^n R^nL T_B^\prime }$.

Clearly, the lower bound for rate $R_2$ follows exactly from Eqs.~(\ref{eq_ffb})-(\ref{eq_ffe}):
\[
R_2\geq \frac{1}{2} I(C;RA).
\] 
To show the lower bound for $R_1$, we begin with
\begin{align}
H(LT_B')_\sigma &+H(M)_\sigma \notag \\
&\geq  H(LT_B'M)_\sigma \label{eq_85} \\
&= H(C_1\widehat{A}^n \widehat{L}\widehat{T}_B' )_\omega \\
&= H(A_1 R^n E' )_\omega \\
&\geq H(A_1)_\tau +H(R^nE')_\theta\label{eq_unc} +f(\epsilon)\\
&\geq H(A^n)_\theta-H(M)_\sigma +H(R^nE')_\theta +f(\epsilon).  \label{eq_lb4}
\end{align} 
The first inequality follows from the subadditivity of the von Neumann entropy. The second inequality (\ref{eq_unc}) follows from Eq.~(\ref{eq_cond21}), the Fannes inequality, and note that $\tau_{A_1}=\tr_{C_1}\Phi_{A_1C_1}$ and $f(\epsilon)\to 0$ as $\epsilon \to0$. The final inequality again follows from the subadditivity of the von Neumann entropy:
\[
H(A^n)_\theta = H(A_1M)_\sigma\leq H(A_1)_\sigma+H(M)_\sigma
\]
and $H(A_1)_\tau\geq H(A_1)_\sigma$ since $\tau_{A_1}$ is the completely mixed state. Then
\begin{align}
R_1&\geq H(M)_\sigma \\
&\geq \frac{1}{2}[H(A^n)_\theta + H(R^nE')_\theta - H(LT_B')_\sigma] \label{eq_lb3} \\
&= \frac{1}{2}[H(A^n)_\theta + H(R^nE')_\theta - H(A^nR^nE')_\theta] \\
&= \frac{1}{2} I(A^n; R^nE')_\theta, \label{eq_94}
\end{align}
where the second line follows from Eq.~(\ref{eq_lb4}), and the third line follows from $H(LT_B')_\sigma= H({L}{T}_B')_\theta= H(A^n R^nE')_\theta$.
These lines (\ref{eq_85})-(\ref{eq_94}) closely follow steps  in the converse proof of the FQSW \cite{Abeyesinghe:2009ej}.

Finally, continuing from Eq.~(\ref{eq_94}) gives
\begin{align}
R_1&\geq \frac{1}{2} I(A^n; R^nE') \\
& = \frac{1}{2} \sum_{i=1}^n I(A_i ; R^n E^\prime | A_{<i}) \\
&= \frac{1}{2}\left[ \sum_{i=1}^n I(A_i ; R^n E^\prime A_{<i}) - I(A_i ; A_{<i}) \right]\\
&=  \frac{1}{2} \sum_{i=1}^n I(A_i ; R^n E^\prime A_{<i}) \\
& \ge \frac{1}{2}\sum_{i=1}^n I(A_i ; R_{\leq i} E^\prime)\\
& =  \frac{1}{2}\sum_{i=1}^n I(A_i ; R_i E_i) \label{eq_lb2} \\
& = \frac{1}{2} n I(A_T ; R_T E_T |T)\label{eq_101} \\
& = \frac{1}{2} n I(A_T ; T R_T E_T) \\
&= \frac{1}{2} n I(A ; R E). \label{eq_lb1}
\end{align}  
To ease the notation, we omit the subscript state $\theta$ in the sequences of equations.  We denote $E_i := (R_{<i},E^\prime)$ in Eq.~(\ref{eq_lb2}), introduce the auxiliary random variable $T$ that is uniformly distributed in the set $\{1,2,\cdots, n\}$ in Eq.~(\ref{eq_101}), and denote $E := (T, E_T)$, $A:=A_T$ and $R:=R_T$ in Eq.~(\ref{eq_lb1}). Note that $E_i$ in Eq.~(\ref{eq_lb2}) is the environment which purifies $U_i$ in Eq.~(\ref{eq_ui}).

\section{Conclusion and Discussion}\label{secIII}

We first considered the problem of compression of a classical source with a quantum helper. We completely characterized its rate region and showed that the capacity formula does not require regularization.  While the expressions for the rate region are similar to the classical result in \cite{Wyner:1975iv, Ahlswede:1975ea, ElGamal:2011ty}, we have employed a different proof technique. To prove the achievability, we employed a powerful theorem, the measurement compression theorem \cite{Winter:2004uk}, that can decompose quantum measurement. A similar approach was recently applied to derive a non-asymptotic bound on  the classical helper problem \cite{Watanabe:2013ea}.

The rate region in our Theorem~\ref{thm_main} bears a close resemblance to its classical counterpart. Our result also shows that a helper's strategy of simply compressing the side information $H(C)_\phi$ and sending it to the decoder is sub-optimal with entanglement assistance. Recall the following identity:
\[
H(C)_\phi = \frac{1}{2} I(C;E)_\phi + \frac{1}{2} I(C;RA)_\phi,
\]
where the state $\ket{\phi_{ACER}}$ is given in Eq.~(\ref{eq_state01}).
The QRST protocol allows us to divide the amount of quantum communication required for lossless transmission of system $C$ to the decoder into pre-shared entanglement with rate $\frac{1}{2}I(C;E)_\phi$ and quantum communication with rate $ \frac{1}{2} I(C;RA)_\phi$.

We would like to point out that the definition of fully quantum source compression with a quantum helper requires explicit inclusion of additional quantum systems $L T_{B}'$ (see Eqs.~(\ref{eq_cond1}) and (\ref{eq_cond21})) for technical purposes. The reason behind this is that when the quantum state merging is performed, the target systems into which the quantum state is merged need to be specified. We believe that the inclusion of these additional systems in the definition is inevitable, and signals a fundamental difference between the fully quantum source compression with a quantum helper and its classical counterpart.  


One interesting direction to extend our results is a problem involving multiple senders and/or multiple receivers (Figure~\ref{fig:GMNSC}). Such a problem has been extensively studied in classical information theory; see \cite{Csiszar:1980ba, TeSunHan:1980ci} and references therein. When there are more than two helpers, such a problem is called two-helper problem, and has not been solved even in the classical case. Thus, it is expected that a quantum version of such a problem is also a difficult problem. When there are multiple receivers, the message sent by one sender may be sent to more than one receiver in the classical setting. However, since quantum states cannot be copied, quantum extension of multiple receiver problems must be  defined carefully. 

Finally, in classical source coding with a helper problem, it is possible to bound the dimension of the helper's output system $|\cU|\leq |\cY|+1$, where $|\cY|$ is the alphabetical size of the random variable $Y$. Such a dimension bound is a simple application of Carath\'eodory theorem.  However, bounding the dimension of an auxiliary
quantum system, i.e. the quantum helper, turns out to be very non-trivial, and very little is known  (cf.~\cite{BeiGoh14}). 



\section*{Acknowledgements}
MH is supported by an ARC Future Fellowship under Grant FT140100574. 
SW was supported in part by JSPS Postdoctoral Fellowships for Research Abroad.

\begin{IEEEbiographynophoto}{Min-Hsiu Hsieh}(M'09--SM'16) received his PhD degree in electrical engineering from the University of Southern California, Los Angeles, in 2008.  From 2008-2010, he was a Researcher at the ERATO-SORST Quantum Computation and Information Project, Japan Science and Technology Agency, Tokyo, Japan. From 2010-2012, he was a Postdoctoral Researcher at the Statistical Laboratory, the Centre for Mathematical Sciences, the University of Cambridge, UK. He is now an Australian Research Council (ARC) Future Fellow and an Associate Professor at the Centre for Quantum Computation \& Intelligent Systems (QCIS), Faculty of Engineering and Information Technology (FEIT), University of Technology Sydney (UTS). His scientific interests include quantum Shannon theory, entanglement theory, and quantum coding theory.
\end{IEEEbiographynophoto}

\begin{IEEEbiographynophoto}
{Shun Watanabe}
(M'09) received the B.E., 
M.E., and Ph.D.\ degrees from the Tokyo Institute of Technology
in 2005, 2007, and 2009, respectively. During April 2009 to February 2015, he was an
Assistant Professor in the Department of Information 
Science and Intelligent Systems at  the University of Tokushima.
During April 2013 to March 2015, he was a visiting Assistant Professor
in the Institute for Systems Research at the University of Maryland, College Park.
Since February 2015, he has been an Associate Professor in the Department of
Computer and Information Sciences at Tokyo University of Agriculture and Technology.
His current research interests are in the areas of
information theory, quantum information theory,
cryptography, and computer science.
\end{IEEEbiographynophoto}

\end{document}